\documentclass{article}
\usepackage{xcolor}
\usepackage{hyperref}
\hypersetup{
  colorlinks   = true, 
  urlcolor     = blue, 
  linkcolor    = green!50!black, 
  citecolor   = red 
}
\usepackage[font=small,labelfont=bf]{caption}
\usepackage{import}
\usepackage{pdfpages}
\usepackage{transparent}

\usepackage{multirow}

\usepackage{authblk}
\usepackage{caption, subcaption}

\usepackage[T2A,T1]{fontenc}    
\usepackage[utf8]{inputenc}     
\usepackage[russian,english]{babel}

\usepackage{mathtools}
\mathtoolsset{showonlyrefs=true}

\DeclareCaptionLabelFormat{subfigure}{\thefigure#2: }
\usepackage{tikz}
\usetikzlibrary {calc}

\usepackage{geometry}

\usepackage{todonotes}
\usepackage{csquotes}

\usepackage{orcidlink}

\usepackage{amsmath,amsfonts,amsthm,amssymb}
\usepackage{bm}

\theoremstyle{definition}
\newtheorem{Definition}{Definition}[section]

\theoremstyle{definition}
\newtheorem{Remark}{Remark}[section]

\theoremstyle{definition}
\newtheorem{Proposition}{Proposition}[section]

\theoremstyle{definition}
\newtheorem{Corollary}{Corollary}[section]

\theoremstyle{definition}

\theoremstyle{definition}
\newtheorem{Lemma}{Lemma}[section]

\theoremstyle{definition}

\theoremstyle{definition}
\newtheorem{Example}{Example}[section]

\theoremstyle{definition}
\newtheorem{Assumption}{Assumption}[section]

\title{Bounds and Phase Transitions for Phonons in Complex Network Structures}

\author[$\dagger$]{Riccardo Bonetto\orcidlink{0000-0001-8075-6147}}

\begin{document}
\affil[$\dagger$]{University of Groningen — Bernoulli Institute for Mathematics, Computer Science and Artificial Intelligence; Groningen, The Netherlands}
\maketitle

\abstract{
We study a model of networked atoms or molecules oscillating around their equilibrium positions. The model assumes the harmonic approximation of the interactions. We provide bounds for the total number of phonons, and for the specific heat, in terms of the average Wiener capacity, or resistance, of the network. Thanks to such bounds, we can distinguish qualitatively different behaviours in terms of the network structure alone.
}





\section{Introduction}

Complex systems show rich and intricate behaviour arising from the interaction of simpler individual constituents. The structure of the interactions can often be described by a network, where each single system is represented by a vertex or node, and the interactions are represented by edges or connections. Complex networks of quantum systems play a relevant role in modern quantum physics \cite{nokkala2023complex}. Examples can be found in quantum communications \cite{bassoli2021quantum, pirandola2019end}, Josephson junctions \cite{nakajima1976logic, niu1989theory}, and condensed matter \cite{alodjants2022phase, dorogovtsev2008critical, herrero2002ising, suchecki2013ising}, among others. A crucial objective is to relate physical quantities to network properties. 

This paper focuses on a central model for solid state physics: atoms, or molecules, oscillating around their equilibrium position, see \eqref{eq:ham}. The collective vibrations give rise to quasiparticles called phonons, whose behaviour determines fundamental properties of condensed matter. The interaction topology is provided by a network structure, therefore generalising classic assumptions, such as lattices, chains and rings. Similar models to the one studied in the present paper have been used as simplified models for networks of nanowires and nanotubes \cite{xiong2018influence, xiong2021regulating}, where numerical simulations have shown the importance of the network structure for the thermodynamic properties. Also, thermal transport properties of phonons in complex networks have been studied in \cite{xiong2024regulation}, while the specific heat of carbon nanotube networks has been investigated in \cite{li2009specific}. Actually, our goal is to contribute to the theoretical understanding of how the underlying network structure is related to the total number of phonons and the specific heat.

The content of the paper is structured as follows. In Section \ref{sec:preliminary}, we decouple the system via some coordinate transformations constructed in terms of algebraic quantities associated with the graph structure. The paper's main results are proved in Section \ref{eq:phase_transitions}. We provide a bound for the average total number of phonons, Proposition \ref{prop:bound_N}, and for the specific heat, Proposition \ref{prop:specific_heat}, in terms of a global network quantity known as average Wiener capacity, or resistance. Additionally, we state conditions discerning between different behaviours in terms of network properties, Proposition \ref{prop:criteria}. In Section \ref{sec:numerical}, we validate numerically the theory developed showcasing the influence of the number of connections of the graph on the thermodynamic quantities, employing a specific class of networks: circulant graphs. We conclude, in Section \ref{sec:discussion}, with a brief discussion on the results obtained and the future directions. In order to keep a smooth flow of ideas, and avoiding too many digressions, we prove some additional results necessary for the main part of the manuscript in a few appendices.


%

\section{Preliminary Results}\label{sec:preliminary}

Let $\mathcal{G}=\{\mathcal{V}, \mathcal{E}\}$ be a simple connected graph, where $\mathcal{V}=\{ 1, \dots, n\}$ is the set of vertices, and $\mathcal{E}$ is the set of edges. We denote by $\{i,j\} \in \mathcal{E}$ an undirected edge connecting the node $i$ to the node $j$, or vice-versa. We consider a quantum system of atoms or molecules (which, from now on, we will refer to as atoms) whose configuration is described by the graph $\mathcal{G}$. In general, among two interacting atoms, say $i$ and $j$, there is an interatomic pair potential $V(q_i-q_j)$. We assume the \emph{harmonic approximation} for the interaction so that the Hamiltonian of the system reads
\begin{equation}\label{eq:ham}
    H= \frac{1}{2} \left( \sum_{i=1}^n p_i^2 + \sum_{\{i,j\} \in \mathcal{E} } (q_i - q_j)^2 \right) ,
\end{equation}
where $q_1, \dots, q_n$ are the positions, and $p_1, \dots, p_n$ are the conjugate momenta. Notice that, without loss of generality, the natural frequency arising from the interaction is set to one, as well as the masses, and the Planck constant $\hbar$. The form of the Hamiltonian \eqref{eq:ham} is a well-known starting point for the study of crystal vibrations \cite{mahan2013many, landau1980lifshitz}. For a long time, the interaction topology considered was restricted to a few particular relevant cases such as lattices, rings, and chains. More recently, the study of complex network structures gave new life to the subject \cite{kim2003netons, burioni2001bose, xiong2024regulation, xiong2018influence, de2007networks}. Indeed, our main goal is to describe the physical properties of the system in terms of the underlying complex network structure.

\subsection{Decoupling and Diagonalisation}

The Hamiltonian \eqref{eq:ham} is fully coupled and, therefore difficult to study in such form. We are going to prove that there are appropriate transformations that decouple and diagonalise \eqref{eq:ham}, leading, in turn, to a more straightforward analysis of the system. The decoupling and diagonalisation of Hamiltonians of the form \eqref{eq:ham} is a standard procedure in almost every book of solid-state physics. However, the main arguments used are based on Fourier methods intrinsically specific to structures such as lattices, chains, and rings. We are going to use properties and tools from algebraic graph theory \emph{applicable to every network structure}. In Appendix \ref{app:graph}, we summarise the key notions of algebraic graph theory necessary for the results in the paper.

Let us start by recalling some basic notions of quantum mechanics. Positions and momenta are Hermitian operators satisfying the canonical commutation relations
\begin{equation}
    [q_i, p_j] = \mathrm{i}  \delta_{ij} , \quad [q_i, q_j] = 0 , \quad [p_i, p_j] = 0,
\end{equation}
where $[\cdot, \cdot]$ denotes the commutator, $\mathrm{i}$ is the imaginary unit, and $\delta_{ij}$ is the Kronecker delta. For each individual system, the Hilbert space is $\mathbb{H}=L^2(\mathbb{C})$, i.e., the space of square-integrable functions. Moreover, for a system with $n$ degrees of freedom, the Hilbert space can be written as the tensor product of the Hilbert spaces of the individual systems, $\mathbb{H}^n := \bigotimes_{i=1}^n \mathbb{H}$. The scalar product on Hilbert spaces is denoted by $\langle \cdot , \cdot  \rangle$, the specific Hilbert space considered will be clear from the context. We also introduce a vectorial notation for position, $q:=(q_1, \dots, q_n)^\intercal$, and momentum, $p:=(p_1, \dots, p_n)^\intercal$. We denote by $*$ the scalar product among components of operators in vectorial form, e.g., $q*p = \sum_{i=1}^n q_i p_i$.

\begin{Proposition}
    Let $S$ be an orthogonal transformation diagonalising the Laplacian matrix $L$, see Lemma \ref{lm:orth},
    \begin{equation}
        S^\intercal L S = \Lambda ,
    \end{equation}
    where $\Lambda=\textup{diag}(\lambda_0, \dots, \lambda_{n-1})$, and $\lambda_0, \dots, \lambda_{n-1}$ are the eigenvalues of $L$. The \emph{canonical transformations}
    \begin{equation}\label{eq:can_tran}
    \begin{aligned}
      q \mapsto S^\intercal q &=: Q  ,\\
        p \mapsto S^\intercal p &=: P ,
    \end{aligned}      
    \end{equation}
    where $Q:=(Q_1, \dots, Q_n)^\intercal$ and $P:=(P_1, \dots, P_n)^\intercal$ are the new position and momentum operators, decouple the Hamiltonian \eqref{eq:ham}. In the new position and momentum operators, the transformed Hamiltonian reads
    \begin{equation}\label{eq:ham_dec}
        H = \frac{1}{2}  \sum_{i=1}^n \left( P_i^2 + \lambda_{i-1} Q_i^2 \right) .
    \end{equation}
\end{Proposition}

\begin{proof}
    Canonical transformations are defined as those transformations that preserve the commutation relations \cite{anderson1994canonical, blaszak2013canonical}. Component-wise, the transformation \eqref{eq:can_tran} reads 
    \begin{equation}
    \begin{aligned}
       Q_i(q,p) &= \sum_{k=1}^n S^\intercal_{ik} q_k  ,\\
       P_i(q,p) &=  \sum_{k=1}^n S^\intercal_{ik} p_k .
    \end{aligned}      
    \end{equation}
    So, we can check the new commutator,
    \begin{align}
        [Q_i, P_j] &= \left(\sum_{k=1}^n S^\intercal_{ik} q_k\right) \left(\sum_{l=1}^n S^\intercal_{jl} p_l\right) - \left(\sum_{l=1}^n S^\intercal_{jl} p_l\right) \left(\sum_{k=1}^n S^\intercal_{ik} q_k\right) \\
                   &= \sum_{k,l=1}^n S^\intercal_{ik}S^\intercal_{jl} \left( q_k p_l - p_l q_k \right) \\
                   &= \sum_{k,l=1}^n S^\intercal_{ik}S_{lj} [q_k,p_l] \\
                   &= \mathrm{i} \sum_{k,l=1}^n S^\intercal_{ik}S_{lj} \delta_{kl} \\
                   &= \mathrm{i} \sum_{k=1}^n S^\intercal_{ik}S_{kj} \\
                   &= \mathrm{i} \delta_{ij} ,
    \end{align}
    where we used the fact that $S^\intercal_{ij}=S_{ji}$, together with its orthogonality. The other commutation relations $[Q_i, Q_j]=[P_i, P_j]=0$ are trivial to check. So, the transformation is canonical.

    Notice that the Hamiltonian \eqref{eq:ham} can be written using the vectorial notation,
    \begin{equation}
        H=\frac{1}{2}  \left(  p*p + q * (Lq) \right) .
    \end{equation}
    Then, applying the transformations \eqref{eq:can_tran}, we obtain
    \begin{align}
         H&=\frac{1}{2}  \left(  \left(S P\right)*\left(S P\right) + \left(S Q\right) * \left(L S Q\right) \right)  \\
            &= \frac{1}{2}  \left(   P * \left( S^\intercal S P\right) + Q * \left(S^\intercal L S Q\right) \right) \\
             &= \frac{1}{2}  \left(   P * P + Q * \left(\Lambda Q\right) \right) .
    \end{align}
    The statement and \eqref{eq:ham_dec} are retrieved.
\end{proof}

\begin{Remark}
    Notice that \eqref{eq:can_tran} has the form of a \emph{classical} canonical transformation. In fact, transformation \eqref{eq:can_tran} is a linear composition of the operators $q_1, \dots, q_n$ and $p_1, \dots, p_n$, respectively, and there is no transformation acting on the individual operators alone.
\end{Remark}

\begin{Proposition}
    Let $P^\textup{tot}$, $a_i$, and $a^\dagger_i$ be the operators defined as follows
    \begin{equation}\label{eq:p_and_a}
        \begin{aligned}
        P^\textup{tot} &:= P_1 , \\
        a_i &:= \frac{1}{\sqrt{2}}\left( \sqrt{\omega_i} Q_{i+1} + \frac{\mathrm{i}}{\sqrt{\omega_i}} P_{i+1} \right) , \\
        a^\dagger_i &:= \frac{1}{\sqrt{2}}\left( \sqrt{\omega_i} Q_{i+1} - \frac{\mathrm{i}}{\sqrt{\omega_i}} P_{i+1} \right) ,
    \end{aligned}
    \end{equation}
    where $\omega_i :=\sqrt{\lambda_i}$, $i=1, \dots, n-1$, are called \emph{frequencies}. Then the Hamiltonian \eqref{eq:ham_dec} can be rewritten in the \emph{diagonal} form
    \begin{equation}\label{eq:ham_diag}
        H= \frac{{P^\textup{tot}}^2}{2}  + \sum_{i=1}^{n-1}  \omega_i \left( a^\dagger_i  a_i + \frac{1}{2} \right).
    \end{equation}
\end{Proposition}

\begin{proof}
    Starting from the Hamiltonian \eqref{eq:ham_diag} and by using the relations \eqref{eq:p_and_a} one obtains \eqref{eq:ham_dec}. 
\end{proof}

    Following from Proposition \ref{prop:laplacian_properties} in Appendix \ref{app:graph}, the first eigenvalue of the Laplacian is zero, $\lambda_0=0$. Then, the operators $a_i$, $a^\dagger_i$, are not well defined for $i=0$. So, for the operator associated with the zero eigenvalue of $L$, we need a specific definition, i.e., $ P^\textup{tot} := P_1$. Notice that $P^\textup{tot}$ is not associated with any oscillation, but it is the total momentum of the system; see Appendix \ref{app:ort} for more details. This is a consequence of the fact that the whole system can freely move in space. One can notice from \eqref{eq:ham_diag} that $P^\textup{tot}$ is a \emph{conserved quantity}. A discussion on the conserved quantities of the system and their relations to different coordinate systems is given in Appendix \ref{app:ort}.

\begin{Corollary}
    The operators $a_i$, $a^\dagger_i$ satisfy the commutation relations
    \begin{equation}\label{eq:com_a_adagger}    
         [a_i,a^\dagger_j] = \delta_{ij} , \quad
         [a_i,a_j] = 0 ,  \quad
         [a^\dagger_i,a^\dagger_j] =0 , 
    \end{equation}
    for all $i,j=1, \dots, n-1$.
\end{Corollary}
\begin{proof}
    By using the expressions \eqref{eq:p_and_a}, and the canonical commutation relations for $Q_i$ and $P_i$, we retrieve the commutation relations \eqref{eq:com_a_adagger}. 
\end{proof}
\begin{Remark}
    Given the relations \eqref{eq:com_a_adagger} we identify $a^\dagger_i$ and  $a_i$, $i=1, \dots, n-1$, with the \emph{creation} and \emph{annihilation} operators of a bosonic system. Therefore, the Hamiltonian \eqref{eq:ham_diag} can be rewritten as
    \begin{equation}
        H= \frac{{P^\textup{tot}}^2}{2}  + \sum_{i=1}^{n-1}  \omega_i \left( N_i + \frac{1}{2} \right) ,
    \end{equation}
    where $N_i:= a^\dagger_i  a_i $ is the \emph{number} operator associated to the $i$-th frequency.
\end{Remark}

Actually, the operator $a^\dagger_i$ creates a \emph{phonon} with frequency $\omega_i$, and vice-versa $a_i$ destroys a phonon with the same frequency. The eigenvalues of $N_i$ are natural numbers representing the number of phonons with frequency $\omega_i$, $i=1, \dots, n-1$. It is also convenient to consider the \emph{centre of momentum} reference frame, i.e., $P^\textup{tot}=0$, where the Hamiltonian reads
    \begin{equation}\label{eq:ham_cm}
        H=  \sum_{i=1}^{n-1}  \omega_i \left( N_i + \frac{1}{2} \right) .
    \end{equation}

\section{Phase Transitions}\label{eq:phase_transitions}

The statistical properties of a system with Hamiltonian \eqref{eq:ham_cm} are the ones of an \emph{ideal Bose gas}; in particular, we have a gas of phonons \cite{pathria2016statistical, srivastava2022physics, srivastava2005statistical}. In thermal equilibrium, for each frequency $\omega_i$ the \emph{average number of phonons} with such a frequency is given by
\begin{equation}\label{eq:av_ni}
    \langle N_i \rangle = \frac{1}{e^{\beta \omega_i} -1} ,
\end{equation}
where $\beta \in \mathbb{R}_+$ is a positive-real parameter physically related to the inverse of the temperature, $\beta= (k_b T)^{-1}$, where $k_b$ is the Boltzmann constant, and it is set to one, $k_b=1$. So, the average of the \emph{total number of phonons} is obtained by summing over all the possible frequencies
\begin{equation}\label{eq:av_N}
    \langle N \rangle =  \sum_{i=1}^{n-1}  \langle N_i \rangle = \sum_{i=1}^{n-1} \frac{1}{e^{\beta \omega_i} -1} .
\end{equation}
From \eqref{eq:av_ni} follows that the average energy of the system with Hamiltonian \eqref{eq:ham_cm} is
\begin{equation}
     \langle H \rangle =  \sum_{i=1}^{n-1} \frac{\omega_i}{e^{\beta \omega_i} -1} +  \sum_{i=1}^{n-1} \frac{\omega_i}{2} .
\end{equation}
Consequently, we can compute the \emph{specific heat} \cite{pathria2016statistical}
\begin{equation}\label{eq:specific_heat}
  c(\beta) = -\beta^2 \frac{\partial \langle H \rangle}{\partial \beta} = \sum_{i=1}^{n-1} (\beta \omega_i)^2 \frac{e^{\beta \omega_i}}{(e^{\beta \omega_i} -1)^2} .
\end{equation}

The two thermodynamic quantities of our interest are the average of the total number of phonons $\langle N \rangle$ and the specific heat, $c(\beta)$. We remark on the importance of these two quantities. We recall that quantum effects become relevant when the quantity 
\begin{equation}\label{eq:density_lambda}
    \frac{  \langle N \rangle}{V} \lambda ,
\end{equation}
where $V$ is the volume occupied by the particles, and $\lambda=h/\sqrt{2\pi m k_b T} = \sqrt{2 \pi \beta}$ is the mean thermal wavelength, is of order one \cite{pathria2016statistical}. Notice that since we are considering particles in a one-dimensional space, the volume is actually a length, and then the wavelength $\lambda$ appears with power one; for particles in a three-dimensional space, one has to adapt the formula using $\lambda^3$. So, roughly speaking, the quantity $\langle N \rangle$ allows us to have an insight on the regime of the system. In turn, this will lead to a different behaviour of the specific heat $c(\beta)$, whose importance is evident. Indeed, historically, the specific heat has been the benchmark for the thermodynamic theory of solids' vibrations \cite{srivastava2005statistical}.



\subsection{Bound for the Total Number of Phonons}

We introduce the concept of \emph{Wiener capacity} for the vertices of a graph, which will be used in later statements.

\begin{Definition}
    Let $\mathbf{e}_1, \dots, \mathbf{e}_n$ be the \emph{standard basis} in $\mathbb{R}^n$, i.e., $(\mathbf{e}_i)_j=\delta_{ij}$. Let $\mathcal{G}= \{ \mathcal{V}, \mathcal{E}\}$ be a connected simple graph. We call \emph{equilibrium measure} of the set $\mathcal{V} \setminus \{i\}$ the unique solution of the linear system
    \begin{equation}
        L \mathbf{v}_i = \mathbf{1} - n \mathbf{e}_i ,
    \end{equation}
    such that $(\mathbf{v}_i)_i=0$. The \emph{Wiener capacity} of the vertex $i$ is defined as
    \begin{equation}
        \textup{cap} (i) := \sum_{j=1}^n (\mathbf{v}_i)_j .
    \end{equation}
\end{Definition}

\begin{Proposition}\label{prop:bound_N}
    The total average number of phonons, $ \langle N \rangle$, in a system with graph structure $\mathcal{G}$ with $n$ vertices, and Hamiltonian \eqref{eq:ham_cm}, satisfies the inequality
    \begin{equation}\label{eq:bound_N}
       \langle N \rangle \leq \alpha_\textup{N}(2) \frac{\overline{\textup{cap}}(\mathcal{G})}{ \beta^2 n} ,
    \end{equation}
    where $\overline{\textup{cap}}(\mathcal{G}) := 1/n \sum_{i=1}^n \textup{cap} (i)$ is the average Wiener capacity of the vertices of $\mathcal{G}$, and $\alpha_\textup{N}(2) \simeq 0.648$, see \eqref{eq:alpha_N}.
\end{Proposition}

\begin{proof}
    We consider the inequality
    \begin{equation}\label{eq:in}
        \frac{1}{e^x -1} \leq  \frac{\alpha_\textup{N}(2)}{x^2} ,
    \end{equation}
    which follows from \eqref{eq:bound_function_N} proven in Appendix \ref{app:bounds}.
    We apply \eqref{eq:in} to the total average number of phonons \eqref{eq:av_N},
    \begin{equation}\label{eq:pippo}
           \langle N \rangle \leq  \alpha_\textup{N}(2) \sum_{i=1}^{n-1} \frac{1}{(\beta \omega_i)^2} = \frac{\alpha_\textup{N}(2)}{\beta^2} \sum_{i=1}^{n-1} \frac{1}{ \lambda_i} ,
    \end{equation}
    where $\lambda_i$, $i=1,\dots, n-1$, are the non-zero eigenvalues of the Laplacian of $\mathcal{G}$. The quantity
    \begin{equation}\label{eq:xiao}
       R (\mathcal{G})= n \sum_{i=1}^{n-1} \frac{1}{ \lambda_i}
    \end{equation}
    is known as \emph{Kirchhoff index}, or \emph{resistance}, of a connected graph \cite{xiao2003resistance}. Moreover there is the following relation between the Kirchhoff index and the Wiener capacity \cite{bendito2008formula, zhu1996extensions, gutman1996quasi}
    \begin{equation}\label{eq:bend}
        R (\mathcal{G}) = \frac{1}{n} \sum_{i=1}^n \textup{cap} (i) = \overline{\textup{cap}}(\mathcal{G}).
    \end{equation}
    Combining together \eqref{eq:pippo}, \eqref{eq:xiao}, and \eqref{eq:bend}, we retrieve the statement.
\end{proof}

\begin{Remark}
    Given \eqref{eq:bend} one can rewrite \eqref{eq:bound_N} in terms of the Kirchhoff index, or resistance, $ \beta^2  \langle N \rangle \leq \alpha_\textup{N}(2)  R (\mathcal{G})/n$. Using the average Wiener capacity has the interesting aspect of relating averages over the quantum statistic to averages over graph properties. Moreover, the following relation between the average Wiener capacity and the Moore-Penrose inverse of the Laplacian matrix, see Corollary \ref{cor:mp_l} and \cite{gutman1996quasi}, holds 
    \begin{equation}\label{eq:cap_vs_Lplus}
        \overline{\textup{cap}}(\mathcal{G}) = n \textup{tr}(L^+) .
    \end{equation}
\end{Remark}

Proposition \ref{prop:bound_N} relates thermodynamic quantities, temperature, and number of phonons to graph properties encoding the topological structure of the system: number of vertices and average Wiener capacity of $\mathcal{G}$. Let us recall that the number of vertices of the graph is equal to the number of atomic (or molecular) sites of the system.
%
%
%
%
%
Network structures having an average Wiener capacity proportional to the number of atomic sites will have a bounded number of phonons regardless of the size of the graph. On the other hand, systems having an average Wiener capacity proportional to a power of the number of vertices have an increasing, and eventually divergent, number of phonons, independently of the temperature. We show a couple of examples to illustrate the two different behaviours.
\begin{Example}
    Let $K_n$ be the complete graph with $n$ vertices. The average Wiener capacity is given by $\overline{\textup{cap}} (K_n) = n-1 $, so we have
\begin{equation}
    \frac{\beta^2  \langle N \rangle}{\alpha_\textup{N}(2) } \leq   1 - \frac{1}{n} ,
\end{equation}
which implies $\beta^2  \langle N \rangle <   \alpha_\textup{N}(2)  $, for all $n = 2, 3, \dots$.
\end{Example}
\begin{Example}
    Let $P_n$ be the path graph with $n$ vertices. The average Wiener capacity is given by $\overline{\textup{cap}} (P_n) = (n-1) n (n+1)/6 $, so we obtain
\begin{equation}
    \frac{\beta^2  \langle N \rangle}{\alpha_\textup{N}(2) } \leq   \frac{n^2-1}{6} ,
\end{equation}
where the bound can become arbitrarily large when the number of atomic sites increases.
\end{Example}
So, there is a phase transition in terms of graph's classes between finite and arbitrarily large bounds on the total average number of phonons. A precise statement is given later in Proposition \ref{prop:criteria}, together with a result for the specific heat. 

Let us recall that the average Wiener capacity is equal to the resistance of a graph. The graph's resistance arose as a concept in the study of electrical networks, where the edges are identified as resistors. As one might expect, the resistance of a graph gives a measure of how much resistance is opposed to the flow of electrical current in the network. Intuitively, we can, therefore, expect that in a graph with many connections among the vertices, the flow of current will face less resistance than in a graph with few connections. This is reflected in what we have seen in the previous examples. The complete graph has a direct connection among any two vertices; as such, the resistance opposed is minimal. On the other hand, for a path graph, the flow of current among two vertices may need to cross many edges, i.e., resistances; therefore, the resistance of the network is greater. Another physical intuition of what is happening comes from the former Hamiltonian \eqref{eq:ham}, which classically can be thought of as a network of harmonic oscillators. In such a context, let us consider the complete graph. Increasing the number of nodes, each vertex will have an increasing number of connections. At some point, the system will become stiffer, and stiffer, as the forces applied to each vertex will reduce the possibility of motion. Conversely, increasing the number of nodes for a loosely connected graph will give just more possible configurations for the oscillations. Quantum mechanically, the considerations we mentioned induce a different behaviour for the number of phonons in the system, which often are referred to as quanta of vibrations or quanta of sound.

\subsection{Bound for the Specific Heat}

In this section, we provide a bound for the specific heat \eqref{eq:specific_heat} at low temperatures, i.e. large $\beta$. Moreover, we state a Lemma exploiting a bifurcation criterion in terms of graph classes. 


\begin{Proposition}\label{prop:specific_heat}
    Let $\mathcal{G}$ be a graph with $n$ vertices, and $\beta$ sufficiently large, i.e. $\beta^2 > \overline{\textup{cap}}(\mathcal{G})/ n $. Then the specific heat  \eqref{eq:specific_heat} of the system with Hamiltonian \eqref{eq:ham_cm}, satisfies the inequality
    \begin{equation}\label{eq:specific_heat_bound}
        c(\beta)  < \alpha_{\textup{E}}  \frac{\frac{\overline{\textup{cap}}(\mathcal{G})}{ \beta^2 n}}{1+\frac{\overline{\textup{cap}}(\mathcal{G})}{ \beta^2 n}} ,
    \end{equation}
    where $\overline{\textup{cap}}(\mathcal{G})$ is the average Wiener capacity of the vertices of $\mathcal{G}$, and $\alpha_{\textup{E}} \simeq 5.23$, see \eqref{eq:alpha_e}.
\end{Proposition}

\begin{proof}
    We use the bound \eqref{eq:bound_Einstein} found in Appendix \ref{app:bounds},
    \begin{equation}
         \frac{x^2 e^x}{(e^x -1)^2} \leq \frac{\alpha_\textup{E}}{x^2+1} .
    \end{equation}
    Then, for the specific heat \eqref{eq:specific_heat}, we obtain the relation
    \begin{equation}
        c(\beta) = \sum_{i=1}^{n-1} (\beta \omega_i)^2 \frac{e^{\beta \omega_i}}{(e^{\beta \omega_i} -1)^2} \leq \sum_{i=1}^{n-1}  \frac{\alpha_\textup{E}}{\beta^2 \lambda_i +1} ,
    \end{equation}
    where $\lambda_1, \dots, \lambda_{n-1}$ are the nonzero eigenvalues of the Laplacian matrix $L$. We collect $\beta^2$, obtaining
    \begin{equation}
        c(\beta)  \leq \frac{\alpha_\textup{E}}{\beta^2} \sum_{i=1}^{n-1}  \frac{1}{ \lambda_i +\beta^{-2}} .
    \end{equation}
    Now, we assume $\beta \gg 1$, and we define $\epsilon:=\beta^{-2}$, so that $\epsilon \ll 1$. Let us recall that for an alternating series we have that $\sum_{k=0}^\infty (-1)^k x^k = (1 + x )^{-1}$, for $|x|<1$. 
    So, for $\epsilon$ sufficiently small, we can write
    \begin{equation}
        \frac{1}{ \lambda_i +\epsilon} = \sum_{k=0}^\infty (-1)^k \frac{\epsilon^k}{\lambda_i^{k+1}} .
    \end{equation}
    We remark that $\epsilon$ sufficiently small means small enough to ensure absolute convergence of the series. Summing over the eigenvalues, we obtain 
    \begin{align}
         \sum_{i=1}^{n-1}  \frac{1}{ \lambda_i + \epsilon} &=   \sum_{i=1}^{n-1}  \sum_{k=0}^\infty (-1)^k \frac{\epsilon^k}{\lambda_i^{k+1}}  \\
         &=    \sum_{k=0}^\infty  \sum_{i=1}^{n-1} (-1)^k \frac{\epsilon^k}{\lambda_i^{k+1}}  \\
          &=    \sum_{k=0}^\infty  (-1)^k \epsilon^k \textup{tr}\left( \left(L^+\right)^{k+1} \right)
    \end{align}
    where $L^+$ is the Moore-Penrose inverse of $L$, see Corollary \ref{cor:mp_l}, and the exchange of the infinite series and the finite sum is possible thanks to the absolute convergence of the series. Since the eigenvalues $\lambda_1, \dots, \lambda_{n-1}$ are all positive we have that 
    \begin{equation}
        \textup{tr}\left( \left(L^+\right)^{k+1} \right) \leq     \left( \textup{tr} \left(L^+\right) \right)^{k+1} ,
    \end{equation}
    where the equality holds only for $k=0$, and for $k>1$ the inequality is strict. Then we have
    \begin{align}
         \sum_{i=1}^{n-1}  \frac{1}{ \lambda_i + \epsilon} < \sum_{k=0}^\infty  (-1)^k \epsilon^k    \left( \textup{tr} \left(L^+\right) \right)^{k+1} =  \textup{tr} \left(L^+\right)  \sum_{k=0}^\infty  (-1)^k     \left( \epsilon \textup{tr} \left(L^+\right) \right)^k .
    \end{align}
    By using \eqref{eq:cap_vs_Lplus}, and setting back $\epsilon=\beta^{-2}$ we have that 
    \begin{equation}
         \sum_{k=0}^\infty  (-1)^k     \left( \epsilon \textup{tr} \left(L^+\right) \right)^k =  \sum_{k=0}^\infty  (-1)^k     \left(  \frac{ \overline{\textup{cap}}(\mathcal{G})}{\beta^2 n} \right)^k . 
    \end{equation}
    At this point, we can give more quantitative estimates on the magnitude of $\beta$ necessary for absolute convergence. Indeed, requiring that the ratio between subsequent terms is strictly less than one, we obtain 
    \begin{align}\label{eq:cond_beta}
        \beta^2 > \frac{ \overline{\textup{cap}}(\mathcal{G})}{ n} .
    \end{align}
    Let us notice that if $\beta$ satisfies \eqref{eq:cond_beta}, then the previous convergence requirements are also satisfied. 
    Finally, using the summation formula for an alternating series we retrieve the statement.
\end{proof}

Proposition \ref{prop:specific_heat} is valid under the assumption of a given graph with a fixed number of vertices, $n$, and a temperature low enough, $\beta^2 > \overline{\textup{cap}}(\mathcal{G})/ n $. Notice that the requirement on $\beta$ is related to the total average number of phonons \eqref{eq:bound_N}. Indeed, it can be read as that the total average number of phonons is small, namely $\langle N \rangle < \alpha_\textup{N}(2)$. In the previous section, we anticipated, with two examples, the occurrence of a phase transition in terms of graph classes for the average number of phonons. Now, we give a rigorous statement that also includes a consequence for the specific heat.

\begin{Proposition}\label{prop:criteria}
    Let $n_1, n_2, \dots$ be a strictly increasing sequence of positive natural numbers. Let $\mathcal{G}_{n_1}, \mathcal{G}_{n_2}, \dots$ be a sequence of graphs, where the labels $n_1 < n_2 < \dots $ are the number of nodes of the respective graph. Then, if $ \overline{\textup{cap}} (\mathcal{G}_{n_j})$ depends linearly on $n_j$ the total average number of phonons, $ \langle N \rangle$ is bounded for all $n_j$, $j = 1, 2, \dots$. Moreover, there exists a $\beta$ such that the specific heat is bounded for all $n_j$, $j = 1, 2, \dots$.
\end{Proposition}

\begin{proof}
   The first part is a direct consequence of Proposition \ref{prop:bound_N}. In particular, it follows from the fact that, if $ \overline{\textup{cap}} (\mathcal{G}_{n_j})$ is a linear function of $n_j$, the quantity
   \begin{equation}
       \frac{\overline{\textup{cap}}(\mathcal{G}_{n_j})}{ \beta^2 n_j}
   \end{equation}
   is finite and bounded for all $j = 1, 2, \dots$, and for all fixed $\beta \in \mathbb{R}_{>0}$. Consequently, from Proposition \ref{prop:specific_heat}, we have that for $\beta^2 > \overline{\textup{cap}}(\mathcal{G})/ n $ there exists a bound for the specific heat. Notice that, given the linearity property, the supremum
   \begin{equation}
       \sup_j \left(\frac{\overline{\textup{cap}}(\mathcal{G}_{n_j})}{ n_j} \right)
   \end{equation}
   is finite. Therefore by choosing $\beta^2 > \sup_j \left(\overline{\textup{cap}}(\mathcal{G}_{n_j}) / n_j \right) $ the bound \eqref{eq:specific_heat_bound} holds, for all $j = 1, 2, \dots$.
\end{proof}

Proposition \ref{prop:criteria} offers a practical bifurcation criterion based on the average Wiener capacity $\overline{\textup{cap}}(\mathcal{G})$ of a graph. By considering the number of vertices of a graph as a parameter, we can distinguish between families of network structures leading to finite bounds, and families leading to arbitrarily large bounds. 
\section{Numerical Validation}\label{sec:numerical}

In this section, we showcase, by numerical simulations, how the behaviour of the total average number of phonons, $\langle N \rangle$, and of the specific heat, $c(\beta)$, depends on the number of vertices, $n$, and connections, for a particular class of regular graphs. We also show a comparison between the bound \eqref{eq:specific_heat_bound} and the values given by the exact formula \eqref{eq:specific_heat} for the specific heat. 
When the dimension of a network becomes large, computing eigenvalues or other algebraic quantities becomes computationally demanding. For such a reason, in order to explore a range of dimension, $n$, that is not restricted to very small graph structures, we focus our numerical study on circulant graphs with $l$ nearest neighbours $Ci(n,l)$, see Appendix \ref{app:circ} for more precise definition and properties. Indeed, in Proposition \ref{prop:eig_circ}, we proved an explicit formula for the eigenvalues, which only depends on $n$ and $l$. Namely, for $j=1, \dots, n-1$, and $2l<n$, the eigenvalues of the Laplacian matrix for $Ci(n,l)$ are
\begin{equation}
    \lambda_j = 2l+1 - \frac{\sin{\left( \frac{j \pi }{n} (2 l+ 1) \right)}}{\sin{\left( \frac{j \pi }{n} \right)}}.
\end{equation}
The class of graphs we are considering allows us to study an important feature, the dependence of the thermodynamic quantities in terms of the number of connections in the graph, that, as we discussed earlier, is related to $\overline{\textup{cap}}(\mathcal{G})$, and therefore to the bounds we provided. In fact, we can choose a number of nearest neighbours that is a function of the number of nodes $l=l(n)$. 

In Figure \ref{fig:behaviour_vs_n}, we show the behaviour of $\langle N \rangle$ and $c(\beta)$, at fixed $\beta=1$, varying with respect to the number of vertices $n$. The number of neighbours is given by $l= \lfloor (n/2)^r \rfloor$, where $r$ is a power between $0$ and $1$, and $\lfloor \cdot \rfloor$ is the floor function. The power $r$ determines the degree of the vertices of the graph. The limit cases $r=0$ and $r=1$ correspond, respectively, to the cycle and complete graphs. We clearly see in Figure \ref{fig:behaviour_vs_n} that by raising the power $r$, and so increasing the network connections, the behaviour changes from increasing and divergent to eventually decreasing and convergent. One can compare the plots in Figure \ref{fig:behaviour_vs_n} with the behaviour of $\overline{\textup{cap}}(\mathcal{G})/n$ in Figure \ref{fig:cap_vs_n}. We recall that the quantity $\overline{\textup{cap}}(\mathcal{G})/n$ determines the bound for $\langle N \rangle$, \eqref{eq:bound_N}, but also the properties of the specific heat, see \eqref{eq:specific_heat_bound} and Proposition \ref{prop:criteria}. By looking at Figure \ref{fig:cap_vs_n}, we notice that the quantity $\overline{\textup{cap}}(\mathcal{G})/n$ provides a qualitative measure for the behaviour of the thermodynamic quantities in Figure \ref{fig:behaviour_vs_n}. Let us also notice that fixing $\beta=1$ does not alter the qualitative behaviour shown. In Figure \ref{fig:c_vs_l}, we show a further example of how the number of connections of a graph affects the specific heat: more connections, lower specific heat.
\begin{figure}[!ht]
     \centering
     \begin{subfigure}[b]{0.8\textwidth}
         \centering
         \includegraphics[width=\textwidth]{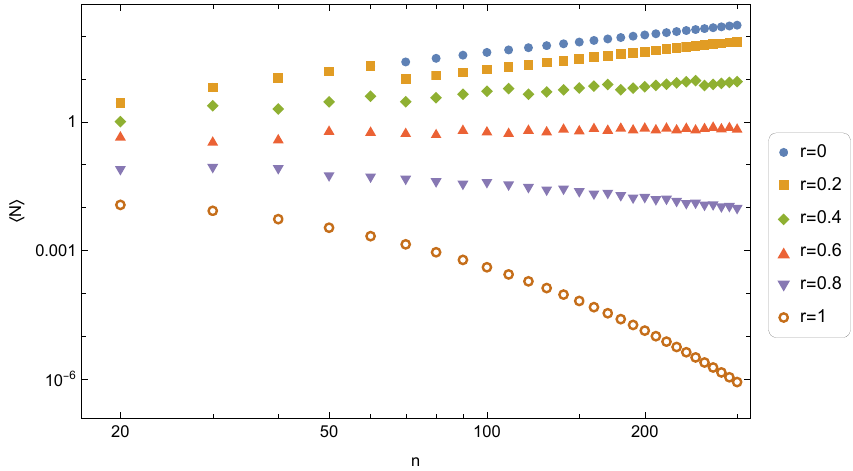}
     \end{subfigure}
     \begin{subfigure}[b]{0.8\textwidth}
         \centering
         \includegraphics[width=\textwidth]{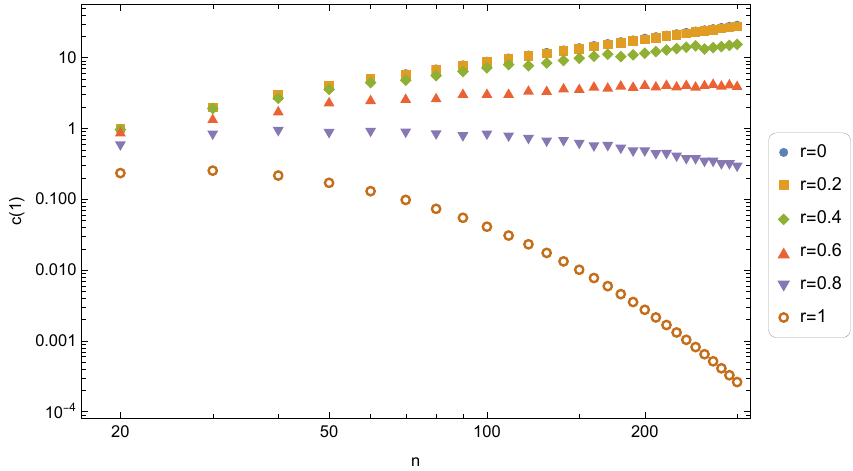}
     \end{subfigure}
        \caption{ Behaviour of $\langle N \rangle$ (top figure) and $c(\beta)$ (bottom figure) with respect to the number of vertices, $n$, of the network. The temperature is fixed, $\beta=1$. We consider circulant graphs, $Ci(n,l)$, with $l= \lfloor (n/2)^r \rfloor$ nearest neighbours. Different powers, $r$, lead to different behaviors in the thermodynamic quantities. Also, the power $r$ provides a proportionality relation between the number of vertices and the number of connections of each vertex. We clearly observe a qualitatively different behaviour depending on the number of connections in the network. Increasing the number of connections causes the patterns to transit from increasing to decreasing.
        Comparing with Figure \ref{fig:cap_vs_n}, we observe that the behaviour displayed can be predicted by looking at the quantity $\overline{\textup{cap}}(\mathcal{G})/n$, solely determined by the network structure.
        }
        \label{fig:behaviour_vs_n}
\end{figure}
\begin{figure}[!ht]
\centering
\includegraphics[width=0.8\textwidth]{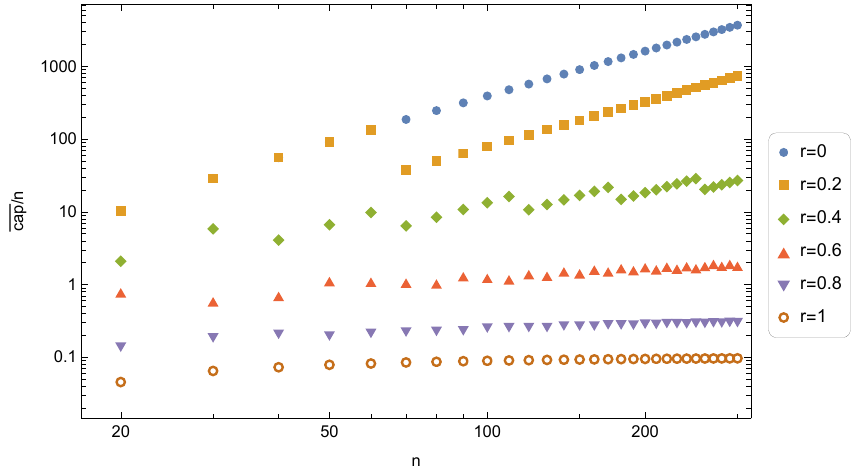}
\caption{ Behaviour of $\overline{\textup{cap}}(\mathcal{G})/n$ for circulant graphs with $l= \lfloor (n/2)^r \rfloor$ nearest neighbours. We notice two main trends depending on the number of connections. For powers $r= 0.6, 0.8, 1$ the quantity $\overline{\textup{cap}}(\mathcal{G})/n$ is converging towards a finite value, while for $r=0.4, 0.2, 0$ it is constantly increasing. Such pattern is in correspondence with what is observed for the thermodynamic quantities $\langle N \rangle$ and $c(\beta)$ in Figure \ref{fig:behaviour_vs_n}. }
\label{fig:cap_vs_n}
\end{figure}
\begin{figure}[!ht]
\centering
\includegraphics[width=0.7\textwidth]{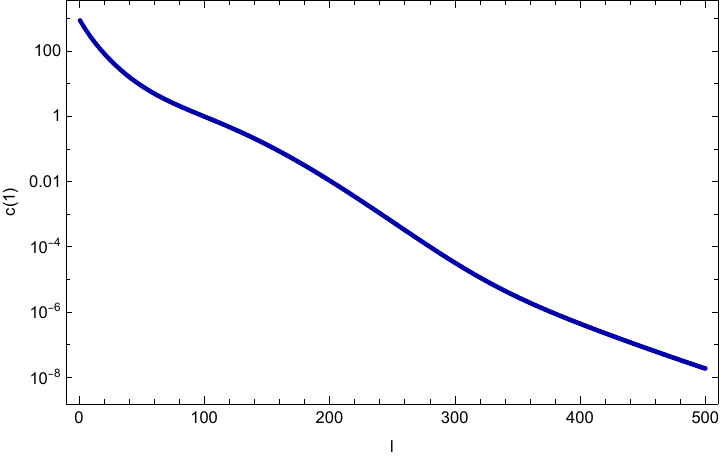}
\caption{ The graphic shows in the vertical axis the specific heat at fixed temperature, $\beta=1$, and in the horizontal axis the number of nearest neighbours, $l$. The number of vertices is fixed $n=1000$, while $l$ ranges from the minimum to the maximum value possible, $l \in [1,500]$. The picture clearly shows that by increasing the number of connections in the network structure the specific heat decrease.    }
\label{fig:c_vs_l}
\end{figure}

In Figure \ref{fig:bounds}, we show a comparison between the bound \eqref{eq:specific_heat_bound} and the specific heat computed with the formula \eqref{eq:specific_heat} as functions of the temperature $T=1/\beta$. We employ circulant graphs $Ci(n,l)$ with $l$ nearest neighbours; the number of vertices, $n$, is fixed. Let us recall that Proposition \ref{prop:specific_heat} is valid for $T^2<n/\overline{\textup{cap}}(\mathcal{G})$. We notice that the limit value $T=\sqrt{n/\overline{\textup{cap}}(\mathcal{G})}$ is always lower than the actual breakdown of the bound. Moreover, both the theoretical and the effective domain of validity of the bound \eqref{eq:specific_heat_bound} increase by increasing $l$.

\begin{figure}[!ht]
     \centering
     \begin{subfigure}[b]{0.45\textwidth}
         \centering
         \includegraphics[width=\textwidth]{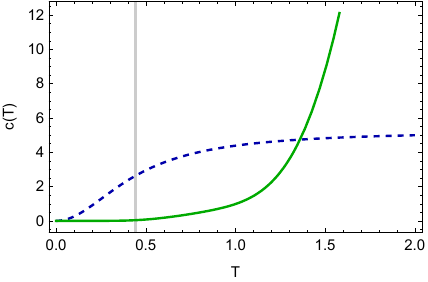}
         \caption{$l=100$ }
         \label{fig:bound_100}
     \end{subfigure}
     \begin{subfigure}[b]{0.45\textwidth}
         \centering
         \includegraphics[width=\textwidth]{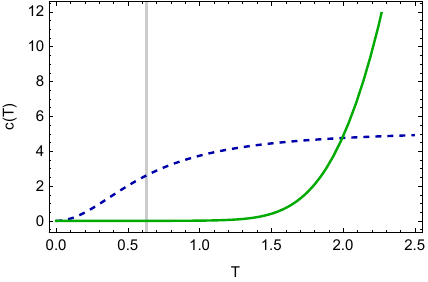}
         \caption{$l=200$}
         \label{fig:bound_200}
     \end{subfigure}
     \centering
     \begin{subfigure}[b]{0.45\textwidth}
         \centering
         \includegraphics[width=\textwidth]{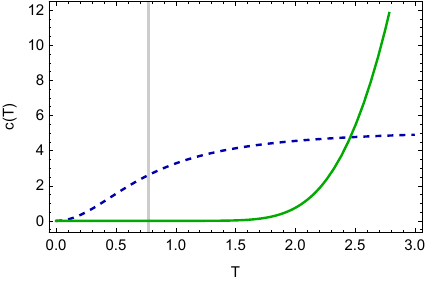}
         \caption{$l=300$ }
         \label{fig:bound_300}
     \end{subfigure}
     \begin{subfigure}[b]{0.45\textwidth}
         \centering
         \includegraphics[width=\textwidth]{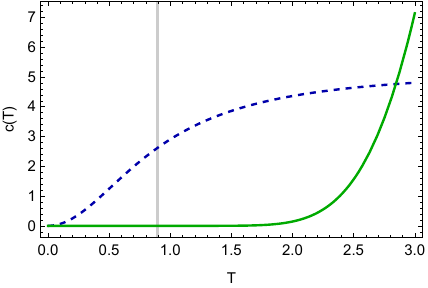}
         \caption{$l=400$}
         \label{fig:bound_400}
     \end{subfigure}
        \caption{ Comparison between the bound \eqref{eq:specific_heat_bound} and the full expression \eqref{eq:specific_heat} for the specific heat as functions of the temperature $T=1/\beta$. We consider circulant graphs $Ci(n,l)$ with fixed number of vertices, $n=1000$, and $l$ nearest neighbours. The four figures show a comparison for different numbers of neighbours $l$. The green-solid line is the specific heat given by the full expression \eqref{eq:specific_heat}, while the blue-dashed line is the bound \eqref{eq:specific_heat_bound}. The vertical line is the upper bound of the domain of validity of \eqref{eq:specific_heat_bound} as given in Proposition \ref{prop:specific_heat}, i.e., $T=\sqrt{n/\overline{\textup{cap}}(\mathcal{G})}$.
        }
        \label{fig:bounds}
\end{figure}

\section{Discussion and Outlook}\label{sec:discussion}

We have shown how the network structure affects the specific heat and the total number of phonons for a quantum system of interconnected elements oscillating around their equilibrium position. Specifically, under appropriate conditions, we proved the existence of a bound for such quantities. Depending on the graph properties, the bounds can be dramatically different. Considering a sequence of graphs labelled by its number of nodes and its number of connections, the specific heat and the total number of phonons can either increase indefinitely or always be finite. This information led us to state the conditions for a phase transition in terms of the network structure, Proposition \ref{prop:criteria}. Indeed, in the numerical validation, we also observed that the bounds, depending on the underlying graph, provide a qualitative measure for the actual behaviour of the thermodynamic quantities. Let us recall that the quantity $ \sqrt{2 \pi \beta} \langle N \rangle  /V $, see \eqref{eq:density_lambda}, can be used as a measure for the regime of the system. If $ \sqrt{2 \pi \beta} \langle N \rangle  /V $ is of order one, the system is in the quantum regime; otherwise, if the order of magnitude is much smaller than one, the system is in the regime of classical mechanics. It is possible to look at transitions between regimes in different ways. One option is by varying the temperature, so for large temperatures, $\beta \ll 1$, the system will behave classically, while at low enough temperatures, $\beta \gg 1$, the system will be in the quantum regime. Given the results of the present paper, we can also see the regime's transitions in terms of network properties. If the average Wiener capacity of the underlying graph is large enough, the system will have a  quantum behaviour, while for a low average Wiener capacity, the system will behave classically. Notice that in the considerations we have made, the volume is always considered fixed and finite.

Several directions can be investigated in the future. The model we studied can be generalised by including the presence of electrons and, therefore of electron-phonon interactions. Such models provide a more realistic account for the thermodynamic properties of solids and also provide a primary mechanism for electron pairing, essential for superconductivity \cite{mahan2013many, pathria2016statistical}. Another option is to consider anharmonic interactions \cite{srivastava2022physics}, in order to understand the interplay between the network structure and anharmonicity. Throughout the paper, we assumed that the network structure was fixed in time. However, it is possible that the interaction topology changes over time, for example, in a timescale much slower than the dynamics described by the Hamiltonian. Interesting phenomena could arise from the multiple timescale framework \cite{kuehn2015multiple, wechselberger2013canard}, e.g. dynamic bifurcations and canard solutions.

\vspace{6pt} 



\appendix

\section{Algebraic Graph Theory}\label{app:graph}

In this appendix we review some basic notions of algebraic graph theory essential for the paper.

\begin{Definition}
    The \emph{adjacency matrix} $A$ of a graph $\mathcal{G}$ is the matrix with components
    \begin{equation}
        A_{ij}:=
        \begin{cases}
            1 \ &\textup{if} \ \{i,j\} \in \mathcal{E} , \\
            0 \ &\textup{otherwise}.
        \end{cases}
    \end{equation}
\end{Definition}

\begin{Definition}
    The \emph{degree} of the node $i \in \mathcal{V}$ is defined as $d_i:= \sum_{j=1}^n A_{ij}$. The matrix $\Delta:=\textup{diag}(d_1, \dots d_n)$ is the \emph{degree matrix} the graph $\mathcal{G}$.
\end{Definition}

\begin{Definition}\label{def:laplacian}
    The matrix $L:= - \Delta + A$ is the \emph{Laplacian matrix} of the graph $\mathcal{G}$.
\end{Definition}


\begin{Proposition}[\cite{godsil2001algebraic}]\label{prop:laplacian_properties}
    The Laplacian matrix $L$ of a connected graph $\mathcal{G}$ is symmetric and positive semidefinite, with exactly one zero eigenvalue associated to the all-one eigenvector $\mathbf{1}=(1, \dots, 1)^\intercal$.
\end{Proposition}

\begin{Lemma}\label{lm:orth}
    There exists an orthogonal transformation $S$, $S^\intercal=S^{-1}$, such that
    \begin{equation}
        S^\intercal L S = \Lambda:=\textup{diag}(0, \lambda_1, \dots, \lambda_{n-1}) .
    \end{equation}
\end{Lemma}

\begin{proof}
    The statement follows form the fact that $L$ is a real symmetric matrix. For a construction see Appendix \ref{app:ort}, and the explicit form \eqref{eq:s}.
\end{proof}

As a consequence of Proposition \eqref{prop:laplacian_properties}, the Laplacian matrix does not admit an inverse. However, a generalisation of inverse matrix can be defined.

\begin{Definition}\label{def:mp_inv}
    Given a real matrix $M$, the \emph{Moore-Penrose inverse}, or \emph{pseudo-inverse}, $M^+$ is a matrix satisfying the conditions
    \begin{align}
        M M^+ M &= M ,\\
        M^+ M M^+ &= M^+ ,\\ 
        (M M^+)^\intercal &= M M^+, \\
        (M^+ M)^\intercal &= M^+ M .
    \end{align}
\end{Definition}

\begin{Corollary}\label{cor:mp_l}
    The \emph{Moore-Penrose inverse}, or \emph{pseudo-inverse}, $L^+$ of the Laplacian matrix is defined as
    \begin{equation}
        L^+ = S \Lambda^+ S^\intercal ,
    \end{equation}
    where $\Lambda^+ = \textup{diag}(0, 1/\lambda_1, \dots, 1/\lambda_{n-1})$.
\end{Corollary}

\begin{proof}
    One can easily check that the conditions of Definition \ref{def:mp_inv} are satisfied (see also \cite{gutman2004generalized}).
\end{proof}
\section{Orthogonalisation and Conserved Quantities}\label{app:ort}

We provide an explicit expression for the orthogonal matrix $S$, which by a similarity transformation brings the Laplacian matrix, $L$, of a graph in diagonal form. Also, we show how to compute explicitly conserved quantities for the quantum system described by the Hamiltonian \eqref{eq:ham} in different coordinate bases.

\begin{Definition}\label{def:gs}
    Let $u$ be a nonzero vector, the \emph{projection} of a vector $v$ along $u$ is defined as
\begin{equation}
    \textup{proj}_u (v) = \left(\frac{v * u}{u * u}\right) u .
\end{equation}
The \emph{Gram–Schmidt orthogonalisation} of the set of vectors $v_0, \dots, v_{n-1}$ is the orthogonal set of vectors $u_0, \dots, u_{n-1}$ obtained as follows
\begin{align}
    u_0 &= v_0 ,\\
    u_i &= v_i - \sum_{j=0}^{i-1}  \textup{proj}_{u_j} (v_i) , 
\end{align}
where $i=1, \dots, n-1$.
The associated \emph{orthonormal basis} is obtained by rescaling to unit length each vector $u_0, \dots, u_{n-1}$,
\begin{equation}
    \Tilde{u}_i := \frac{u_i}{\sqrt{u_i * u_i}} , 
\end{equation}
where $i=0, \dots, n-1$.
\end{Definition}

\begin{Corollary}
 Let $v_0, \dots, v_{n-1}$ be a set of eigenvectors of the Laplacian matrix $L$, associated to the eigenvalues $\lambda_0, \dots, \lambda_{n-1}$. Then the orthogonal transformation $S$ diagonalising $L$ is the matrix having column vectors  $\Tilde{u}_0, \dots, \Tilde{u}_{n-1}$ obtained by a Gram–Schmidt orthogonalisation and normalisation to unit length, which in matrix form reads
\begin{equation}\label{eq:s}
    S=
    \begin{pmatrix}
        (\Tilde{u}_0)_1 & \cdots & (\Tilde{u}_{n-1})_1 \\
        \vdots & \ddots & \vdots \\
        (\Tilde{u}_0)_{n} & \cdots & (\Tilde{u}_{n-1})_n
    \end{pmatrix} .
\end{equation}   
\end{Corollary}

\begin{proof}
    The statement follows from Definition \ref{def:gs}.
\end{proof}

Let $\lambda_0=0$ be the first eigenvalue of $L$ and $v_0=(1, \dots, 1)^\intercal$ its eigenvector. From Definition \ref{def:gs} we have that
\begin{equation}
    \Tilde{u}_0 = \frac{1}{\sqrt{n}} (1, \dots, 1)^\intercal .
\end{equation}
As a consequence we have that
\begin{align}
    P_1 &= \sum_{j=1}^n S^\intercal_{1j} p_j \\
        &= \sum_{j=1}^n S_{j1} p_j \\
        &= \Tilde{u}_0 * p \\
        &= \frac{1}{\sqrt{n}} (p_1 + \dots + p_n) .
\end{align}
So $P_1$, up to a rescaling, is the total momentum of the system. Furthermore, $P_1$ is an example of conserved quantity. Such property follows from a more general result that we provide below.

\begin{Corollary}\label{cor:cons}
    The Hermitian operators
    \begin{equation}\label{eq:hi}
        H_i = P_i^2 + \lambda_{i-1} Q_i^2 ,
    \end{equation}
    where $i=1, \dots, n $, form a set of $n$ commuting \emph{observables},
    \begin{equation}\label{eq:commuting_energies}
        [H_i,H_j] = 0 , 
    \end{equation}
    for all $ i,j=1, \dots n$.
    Moreover, in the coordinates $q$ and $p$ the corresponding observables are given by
    \begin{equation}\label{eq:hi_qp}
        H_i = (\Tilde{u}_{i-1} * p)^2 + \lambda_{i-1} (\Tilde{u}_{i-1} * q)^2 , \quad i=1, \dots, n,
    \end{equation}
    where $\Tilde{u}_0, \dots, \Tilde{u}_{n-1}$ is an orthonormal basis of the Laplacian matrix.
\end{Corollary}

\begin{proof}
    Let us notice that each quantity $H_i$ depends only on $Q_i$ and $P_i$, $i=1, \dots, n$. Then, the commutation relations \eqref{eq:commuting_energies} follow from the canonical commutation relations of $Q$ and $P$. The expression \eqref{eq:hi_qp} in terms of $q$ and $p$ is obtained by using the inverse transformations $Q=S^\intercal q$, $P=S^\intercal p$, together with expression \eqref{eq:s} for the matrix $S$.
\end{proof}



\section{Circulant Graphs}\label{app:circ}

The study of circulant graphs is motivated by the possibility of performing explicit computations thanks to the properties of their associated algebraic quantities \cite{monakhova2012survey}. The structure of a circulant graph is given by a ring of nodes connected to each other by a neighbouring rule. Well-known examples are the cycle graph, or ring network, and the complete graph, or all-to-all network. In between these two examples there is a whole zoo of possibilities provided by different combinations of connections among the set of nodes. From the physical perspective, circulant graphs have the properties that the associated eigenvalues and eigenvectors can be written in terms of Fourier modes.

\begin{Definition}\label{def:circ}
    Let $\mathcal{V}= \{1, \dots, n\}$ be the set of vertices, and let $\mathcal{N}$ be a list such that $\mathcal{N} \subseteq \{ 1, 2 , \dots, \lfloor n/2 \rfloor \}$, where $\lfloor \cdot  \rfloor$ denotes the floor function.
    The \emph{circulant graph} $Ci(n, \mathcal{N})$ is the graph with $n$ vertices such that each vertex $i$ is adjacent to the $\textup{mod}_n (i+j)$-th and $\textup{mod}_n(i-j)$-th vertices, for all $j \in \mathcal{N}$.
\end{Definition}

\begin{Remark}
    If $\mathcal{N}=\{ 1, 2 , \dots, \lfloor n/2 \rfloor \}$ then $Ci(n, \{ 1, 2 , \dots, \lfloor n/2 \rfloor \})$ is the \emph{complete graph} with $n$ vertices. Another example is $Ci(n, \{ 1\})$ which is the \emph{cycle} graph with $n$ vertices.
\end{Remark}

\begin{Definition}
    A \emph{real circulant matrix}, $C$, is a matrix of the form
    \begin{equation}\label{eq:circulant}
    C=
        \begin{pmatrix}
            c_0 & c_{n-1}  & \cdots  & c_1 \\
            c_1 & c_0  & \cdots &  c_2 \\ 
             \vdots & \vdots  & \ddots  &  \vdots \\ 
             c_{n-1} & c_{n-2}  & \cdots  & c_0
        \end{pmatrix} ,
    \end{equation}
    where $c_0, \dots, c_{n-1} \in \mathbb{R}$.
\end{Definition}

\begin{Corollary}[\cite{gray2006toeplitz}]\label{cor:eig_circulant}
    The eigenvalues $\mu_0, \dots, \mu_{n-1}$ and eigenvectors $w_0, \dots, w_{n-1}$ of a circulant matrix \eqref{eq:circulant} are
    \begin{equation}\label{eq:eig_circ}
        \begin{aligned}
      \mu_j &= c_0 + c_1 f^j + c_2 f^{2j} + \dots + c_{n-1} f^{(n-1) j} , \\
        w_j &= \frac{1}{\sqrt{n}} \left( 1 , f^j, f^{2j}, \dots, f^{(n-1) j} \right)^\intercal ,
    \end{aligned}
    \end{equation}
    where $f=\exp(2 \pi \mathrm{i}/n)$ is the $n$-th root of unity, and $j=0, \dots, n-1$.
\end{Corollary}

As a consequence of Definition \ref{def:circ} we have that the adjacency and Laplacian matrix of a circulant graph is a symmetric circulant matrix. We immediately see the great advantage of studying such a class of graphs: from Corollary \ref{cor:eig_circulant} we have an explicit formula for eigenvectors and eigenvalues. We further specialise our study by making the following assumption.

\begin{Assumption}
 We consider the class of circulant graphs with $l$, $2l<n$, nearest neighbours, $Ci(n,l):=Ci(n, \{1, \dots, l\})$.
\end{Assumption}

The requirement $2l<n$ allows us to avoid some complications arising from the fact that if a circulant graph has exactly $\lfloor n/2 \rfloor$ nearest neighbours there is the possibility that the $(i+\lfloor n/2 \rfloor)$-th and $(i-\lfloor n/2 \rfloor)$-th neighbours of the $i$-th vertex coincide. We recall that the case $2l=n$ is the complete graph with $n$ vertices, whose properties are well-known.

\begin{Proposition}\label{prop:eig_circ}
    The eigenvalues of the Laplacian matrix of a circulant graph $Ci(n,l)$ are
    \begin{align}
        \lambda_0 &= 0 , \\
        \lambda_j &= 2l+1 - \frac{\sin{\left( \frac{j \pi }{n} (2 l+ 1) \right)}}{\sin{\left( \frac{j \pi }{n} \right)}} , 
    \end{align}
    where $j=1, \dots, n-1$.
\end{Proposition}

\begin{proof}
    
    We notice that the the first row, or collum, of the Laplacian matrix of $Ci(n,l)$ is
    \begin{equation}\label{eq:row_l}
        (2l, \underbrace{-1, \dots, -1}_{l}, \underbrace{0,\dots,0}_{n-2l-1}, \underbrace{-1, \dots, -1}_{l}) .
    \end{equation}
    The components of \eqref{eq:row_l} are the terms $c_0, \dots, c_{n-1}$ appearing in the formula for the eigenvalues \eqref{eq:eig_circ}. Let us explicitly compute the eigenvalues,
    \begin{align}\label{eq:eig_sums}
        \lambda_j = 2l - \sum_{s=1}^l \exp{\left( \frac{2 \pi \mathrm{i} j}{n} s \right)} - \sum_{s=n-l}^{n-1} \exp{\left( \frac{2 \pi \mathrm{i} j}{n} s \right)}. 
    \end{align}
    For $j=0$ it is easy to see that $\lambda_0=0$, so we are left to compute $\lambda_j$, $j=1, \dots, n-1$. Notice that the second sum in \eqref{eq:eig_sums} can be reordered as follows,
    \begin{align}
        \sum_{s=n-l}^{n-1} \exp{\left( \frac{2 \pi \mathrm{i} j}{n} s \right)}&= \sum_{s=1}^l \exp{\left( \frac{2 \pi \mathrm{i} j}{n} (n-s) \right)} \\
        &= \sum_{s=1}^l \exp{\left( 2 \pi \mathrm{i} j - \frac{2 \pi \mathrm{i} j}{n} s \right)} \\
         &= \sum_{s=1}^l \exp{\left( 2 \pi \mathrm{i} j \right)} \exp{\left(-\frac{2 \pi \mathrm{i} j}{n} s \right)} \\
          &= \sum_{s=1}^l \exp{\left(-\frac{2 \pi \mathrm{i} j}{n} s \right)} .
    \end{align}
    Then the expression for the eigenvalues becomes
    \begin{align}
        \lambda_j &= 2l - \sum_{s=1}^l \exp{\left( \frac{2 \pi \mathrm{i} j}{n} s \right)} - \sum_{s=1}^l \exp{\left(-\frac{2 \pi \mathrm{i} j}{n} s \right)} \\ 
        &= 2l - \sum_{s=1}^l \left( \exp{\left( \frac{2 \pi \mathrm{i} j}{n} s \right)} + \exp{\left(-\frac{2 \pi \mathrm{i} j}{n} s \right)} \right) \\
        &=2l - 2\sum_{s=1}^l  \cos{\left( \frac{2 \pi j}{n} s \right)} . \label{eq:sum_almost_end}
    \end{align}
    Now we use the following result for summations of cosines \cite{jolley2012summation}
    \begin{equation}\label{eq:sum_cos}
        \sum_{s=1}^l \cos{(\theta s)} = \frac{\cos{\left(\frac{(l+1)\theta}{2}\right)} \sin{\left( \frac{l \theta}{2} \right)}}{\sin{\left( \frac{\theta}{2} \right)}} .
    \end{equation}
    Substituting $\theta=\frac{2 \pi j}{n}$ the numerator becomes
    \begin{equation}\label{eq:trig_id}
        \cos{\left(\frac{\pi j}{n} (l+1)\right)} \sin{\left( \frac{\pi j}{n} l \right)} = \frac{1}{2} \left( \sin{\left(\frac{\pi j}{n} (2l+1)\right)} - \sin{\left(\frac{\pi j}{n}\right)} \right) ,
    \end{equation}
    where we used the trigonometric identity $\cos{\theta}\sin{\phi} =  1/2 (\sin{(\theta+\phi)} - \sin{(\theta-\phi)})$. Then, collecting together \eqref{eq:sum_cos} and \eqref{eq:trig_id}, the sum in \eqref{eq:sum_almost_end} becomes
    \begin{equation}
        2\sum_{s=1}^l  \cos{\left( \frac{2 \pi j}{n} s \right)} = -1 + \frac{\sin{\left( \frac{j \pi }{n} (2 l+ 1) \right)}}{\sin{\left( \frac{j \pi }{n} \right)}} ,
    \end{equation}
    which proves the statement.
\end{proof}
\section{Global Bounds for Thermodynamic Functions}\label{app:bounds}

In this appendix, we provide a bound for the function associated to the average number of particles \eqref{eq:av_ni},
\begin{equation}
    \frac{1}{e^x-1} ,
\end{equation}
and a bound for the function related to the summands of the specific heat \eqref{eq:specific_heat}
\begin{equation}\label{eq:Einstein_f}
    \frac{x^2 e^x}{(e^x -1)^2}.
\end{equation}
The function \eqref{eq:Einstein_f} is also known as \emph{Einstein function} \cite{pathria2016statistical}.


\begin{Proposition}
    Let $k=2, 3, \dots$. Then the following inequality holds
    \begin{equation}\label{eq:bound_function_N}
         \frac{1}{e^x - 1} \leq \frac{\alpha_\textup{N}(k)}{x^k} ,
    \end{equation}
    where the constant $\alpha_\textup{N}(k)$ is given by
    \begin{equation}\label{eq:alpha_N}
       \alpha_\textup{N}(k):= -W\left(-e^{-k} k\right) \left(W\left(-e^{-k} k\right)+k\right)^{k-1} ,
    \end{equation}
    where $W$ is the Lambert W function.
\end{Proposition}

\begin{proof}
    We look for the maximum of the function
    \begin{equation}\label{eq:max_for_n}
         \frac{x^k}{e^x - 1}.
    \end{equation}
    So, the first-order derivative is 
    \begin{equation}
        \frac{\left(e^x (k-x)-k\right) x^{k-1}}{\left(e^x-1\right)^2} .
    \end{equation}
    The positive root, obtained by solving $\left(e^x (k-x)-k\right)=0$, is $x_* = W\left(-e^{-k} k\right)+k$. Substituting $x_*$ in \eqref{eq:max_for_n}, and simplifying, we obtain $-W\left(-e^{-k} k\right) \left(W\left(-e^{-k} k\right)+k\right)^{k-1}$. 
\end{proof}

\begin{Remark}
    Some approximate values for $\alpha_\textup{N}(k)$ are $\alpha_\textup{N}(2) \simeq 0.648$, $\alpha_\textup{N}(3) \simeq 1.421$, $\alpha_\textup{N}(4) \simeq 4.780$, $\alpha_\textup{N}(5) \simeq 21.201$.
\end{Remark}





\begin{Proposition}
    The Einstein function \eqref{eq:Einstein_f} satisfy the inequality
    \begin{equation}\label{eq:bound_Einstein}
         \frac{x^2 e^x}{(e^x -1)^2} \leq \frac{\alpha_\textup{E}}{x^2+1},
    \end{equation}
    where $\alpha \simeq 5.23$ is a constant given by
    \begin{equation}\label{eq:alpha_e}
        \alpha_\textup{E} := \frac{(x_*^2+1)x_*^2 e^{x_*}}{(e^{x_*} -1)^2} ,
    \end{equation}
    where $x_*$ is the positive root of the function
    \begin{equation}
       x+2+ (x+4) x^2+e^x \left(x^3-4 x^2+x-2\right) .
    \end{equation}
\end{Proposition}

\begin{proof}
    In order to find the bound \eqref{eq:bound_Einstein}, we look for the maximum of
    \begin{equation}\label{eq:function_to_max_Ein}
        \frac{(x^2+1)x^2 e^{x}}{(e^{x} -1)^2} .
    \end{equation}
    As we can see in Figure \ref{fig:maxEin}, the function \eqref{eq:function_to_max_Ein} has a minimum at zero, and one maximum for positive $x$. We recall that we are only interested in positive $x$, as it is related to the frequencies and temperature of the system, i.e., positive quantities. The first-order derivative of \eqref{eq:function_to_max_Ein} is
    \begin{equation}
        -\frac{e^x x \left((x+4) x^2+e^x \left(x^3-4 x^2+x-2\right)+x+2\right)}{\left(e^x-1\right)^3} .
    \end{equation}
    Excluding the trivial zero at the origin, we have that the maximum we are looking for is a zero of 
    \begin{equation}\label{eq:function}
        g(x):=x+2+ (x+4) x^2+e^x \left(x^3-4 x^2+x-2\right) ,
    \end{equation}
    in the domain of positive $x$. We prove the existence and uniqueness of such a root.
    In order to have a zero of \eqref{eq:function} the polynomial
    \begin{equation}\label{eq:int_pol}
        x^3-4 x^2+x-2 
    \end{equation}
    must be negative, following from the fact that the factor of \eqref{eq:int_pol} and the other quantities in \eqref{eq:function} are all positive for $x>0$. Solving the polynomial equation \eqref{eq:int_pol} we have that the zero of \eqref{eq:function}, must be in the range
    \begin{equation}
        0<x<\frac{1}{3} \left(\sqrt[3]{6 \sqrt{87}+73}+\sqrt[3]{73-6 \sqrt{87}}+4\right)=:s_0.
    \end{equation}
    To show the uniqueness of the root for $x>0$, we compute the first four derivatives of \eqref{eq:function},
    \begin{align}
       g'(x) &= x (3 x+8)+e^x (x ((x-1) x-7)-1)+1 ,\\
       g''(x) &= 6 x+e^x (x (x (x+2)-9)-8)+8 ,\\
       g'''(x) &= e^x (x (x (x+5)-5)-17)+6 ,\\
       g''''(x) &= e^x (x (x (x+8)+5)-22) .
    \end{align}
    We notice that $g''''$ has only one positive zero at 
    \begin{equation}
        s4:= -\frac{8}{3}+\frac{7}{3} \cos \left(\frac{1}{3} \left(\pi -\arctan\left(\frac{6 \sqrt{66}}{5}\right)\right)\right)+\frac{7}{3} \cos \left(\frac{1}{3} \left(\arctan\left(\frac{6 \sqrt{66}}{5}\right)-\pi \right)\right) .
    \end{equation}
    Then, we have that the third-order derivative has a minimum (one can check that by straightforward computations) at $s_4$, with $g'''(s_4)<0$. So, before and after the minimum the function $g'''(x)$ is monotonic. Since $g'''(s_0)>0$ and  $g'''(0)<0$, there exist a unique zero, $ g'''(s_3)=0$, which lies in the interval $s_4<s_3<s_0$. A similar argument is repeated for the other derivatives. After four steps we obtain the conclusion that the function $g(x)$ has a unique minimum for positive $x$, leading to a unique root for $x>0$ thanks to the arguments based on monotonicity.
\begin{figure}[!ht]
\centering
\includegraphics[width=0.4\textwidth]{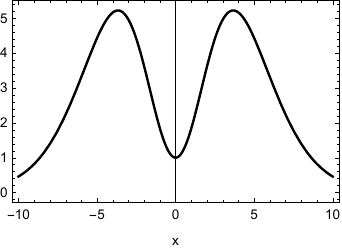}
\caption{Graphic of the function \eqref{eq:function_to_max_Ein}.}
\label{fig:maxEin}
\end{figure}
\end{proof}


\bibliography{bibliography.bib}
\bibliographystyle{abbrv}

\section*{Statements and Declarations}

\begin{description}
    \item[\textbf{Funding}] The authors declare that no funds, grants, or other support were received during the preparation of this manuscript.
    \item[\textbf{Competing Interests}] The authors have no relevant financial or non-financial interests to disclose.
    \item[\textbf{Data Availability}] All datasets are available upon request to the corresponding author.
\end{description}

\end{document}